\documentclass{llncs}

\usepackage[ruled,vlined,linesnumbered,commentsnumbered]{algorithm2e}

\usepackage{amsmath,amssymb,amsfonts}

\input{epsf}

\newcommand {\ignore} [1] {}

\begin{document}

\title{A $1.75$ LP approximation for the \\ Tree Augmentation Problem}

\author{
Guy Kortsarz\inst{1} 
\and 
Zeev Nutov\inst{2}} 
\institute{
Rutgers University--Camden, NJ. \email{guyk@camden.rutgers.edu}
\and The Open University of Israel. \email{nutov@openu.ac.il.}}

\maketitle

\begin{abstract} 
In the {\sf Tree Augmentation Problem} ({\sf TAP}) the goal is to augment a tree $T$ by a minimum size 
edge set $F$ from a given edge set $E$ such that $T \cup F$ is $2$-edge-connected.
The best approximation ratio known for {\sf TAP} is $1.5$. 
In the more general {\sf Weighted TAP} problem, $F$ should be of minimum weight.
{\sf Weighted TAP} admits several $2$-approximation algorithms w.r.t. to the standard cut LP-relaxation,
but for all of them the performance ratio of $2$ is tight even for {\sf TAP}.
The problem is equivalent to the problem of covering a laminar set family.
Laminar set families play an important role in the design of approximation algorithms for 
connectivity network design problems. 
In fact, {\sf Weighted TAP} is the simplest connectivity network design
problem for which a ratio better than $2$ is not known.
Improving this ``natural'' ratio is a major open problem,
which may have implications on many other network design problems.
It seems that achieving this goal requires finding an LP-relaxation with integrality 
gap better than $2$, which is a long time open problem even for {\sf TAP}.
In this paper we introduce such an LP-relaxation and give an algorithm that 
computes a feasible solution for {\sf TAP} of size at most $1.75$ times the optimal LP value. 
This gives some hope to break the ratio $2$ for the weighted case. 
Our algorithm computes some initial edge set by solving a partial system of constraints 
that form the integral edge-cover polytope, and then applies local search on $3$-leaf subtrees 
to exchange some of the edges and to add additional edges.  
Thus we do not need to solve the LP, 
and the algorithm runs roughly in time required to find a minimum weight edge-cover in a general graph. 
\end{abstract}

\section{Introduction} \label{s:intro}

\subsection{Problem definition and related problems}

A graph (possibly with parallel edges) is {\em $k$-edge-connected} if there are $k$ pairwise
edge-disjoint paths between every pair of its nodes.
We study the following fundamental connectivity augmentation problem:
given a connected undirected graph $G=(V,{\cal E})$ 
and a set of additional edges (called ``links") $E$ on $V$ disjoint to ${\cal E}$, 
find a minimum size edge set $F \subseteq E$ such that $G+F=(V,{\cal E} \cup F)$ is $2$-edge-connected.
The $2$-edge-connected components of the given graph $G$ form a tree.
It follows that by contracting these components, one may assume that
$G$ is a tree.  Hence, our problem is:

\begin{center} \fbox{\begin{minipage}{0.965\textwidth} \noindent
{\sf Tree Augmentation Problem (\sf TAP)} \\
{\em Instance:}  \ A tree $T=(V,{\cal E})$ and a set of links $E$ on $V$ disjoint to ${\cal E}$. \\
{\em Objective:} Find a minimum size subset $F \subseteq E$ of links such that $T \cup F$ is $2$-edge-connected.
\end{minipage}}\end{center}

{\sf TAP} can be formulated as the problem of covering a laminar set family as follows.
Root $T$ at some node $r$.
Every edge of $T$ partitions $T$ into two parts $T'$ and $T \setminus T'$, where $r \notin T'$;
let $\hat{\cal E}$ denote the set family obtained by picking for each edge the part $T'$ that does not contain $r$.
Then $\hat{\cal E}$ is laminar, and $F \subseteq E$ is a feasible solution for {\sf TAP} if and only if 
$F$ covers $\hat{\cal E}$, namely, for every $T' \in \hat{{\cal E}}$ there is a link in $F$ from $T'$ to $T \setminus T'$. 
{\sf TAP} is also equivalent to the problem of augmenting the edge-connectivity
from $k$ to $k+1$ for any odd $k$;
this is since the family of minimum cuts of a $k$-connected graph with $k$ odd is laminar.

In the more general {\sf Weighted TAP} problem, 
the links in $E$ have weights $\{w_e:e \in E\}$ and the goal is to find a minimum weight augmenting 
edge set $F \subseteq E$ such that $T \cup F$ is $2$-edge connected.
Even a more general problem is the {\sf $2$-Edge-Connected Subgraph} problem,
where the goal is to find a spanning $2$-edge-connected subgraph of a given weighted graph;
{\sf Weighted TAP} is a particular case, when the input graph contains a connected spanning subgraph of 
cost zero.

\subsection{Previous and related work}

{\sf TAP} is NP-hard even for trees of diameter $4$ \cite{FJ},
or when the set $E$ of links forms a cycle on the leaves of $T$ \cite{CJR}.
The first $2$-approximation for {\sf Weighted TAP} was given 24 years ago in 1981 by 
Fredrickson and J\'{a}j\'{a} \cite{FJ}, and was simplified later by Khuller and Thurimella \cite{KT}.
These algorithms reduce the problem to the {\sf Min-Cost Arborescence} problem, 
that is solvable in polynomial time \cite{CL,E}.
The primal-dual algorithm of \cite{GW,GGPS} is another combinatorial $2$-approximation algorithm for the problem. 
The iterative rounding algorithm of Jain \cite{Jain} is an LP-based $2$-approximation algorithms.
The approximation ratio of $2$ for all these algorithms is tight even for {\sf TAP}.
These algorithms achieve ratio $2$ w.r.t. to the standard {\em cut LP} that seeks 
to minimize $\sum_{e \in E} w_ex_e$ over the following polyhedron:
\begin{equation*} 
\left\{x \in \mathbb{R}^E: x(\delta(A)) \geq 1 \ \forall A \in \hat{{\cal E}}, 
x_e \geq 0 \ \forall e \in E \right\} \ .
\end{equation*}
Here $\delta(A)$ is the set of links with exactly one endnode in $A$ and $x(F)=\sum_{e \in F} x_e$
is the sum of the variables indexed by the links in $F$.

Laminar set families play an important role in the design and analysis of exact and approximation algorithms for 
network design problems, both in the primal-dual method and the iterative rounding method, c.f. \cite{LRS,GW}. 
{\sf Weighted TAP} is the simplest network design problem for which a ratio better than $2$ is not known.
Breaking the ``natural'' ratio of $2$ for {\sf Weighted TAP} is a major open problem in network design,
which may have implications on many other problems, c.f. the surveys \cite{K,KN}.

As a starting point, Khuller \cite{K} in his survey on high connectivity network design problems 
posed as a major open question achieving ratio better than $2$ for {\sf TAP}.
This open question was resolved by Nagamochi \cite{N}, that used a novel lower bound
to achieve ratio $1.875+\epsilon$ for {\sf TAP}.
Building on the lower bound idea of Nagamochi \cite{N}, the sequence of papers \cite{EFKN-APPROX,EFKN-TALG,KN-TAP} 
introduced additional new techniques to achieve ratio $1.8$ by a much simpler algorithm and analysis,
and also achieved the currently best known ratio $1.5$ by a more complex algorithm.

Several algorithms for {\sf Weighted TAP} with ratio better than $2$ are known for special cases.
Cheriyan, Jord\'{a}n, and Ravi \cite{CJR} 
showed how to round a half-integral solution to the cut LP within ratio $4/3$.
However, as is pointed in \cite{CJR}, 
there are {\sf TAP} instances that do not have an LP optimal solutions which is half integral.
In \cite{CN} is given an algorithm with ratio $(1+\ln 2)$ and running time $n^{f(D)}$ where $D$ is 
the diameter of $T$. 

Studying various LP-relaxations for {\sf TAP} is motivated by the hope
that these may lead to breaking the ratio of $2$ for {\sf Weighted TAP}.
Thus several paper ana\-lyzed integrality gaps of LP relaxations for the problem.
Cheriyan, Karloff, Khandekar, and Koenemann \cite{CKKK} showed that the integrality gap of the 
standard cut LP is at least $1.5$ even for {\sf TAP}.
For the special case of {\sf TAP} when every link connects two leaves, 
\cite{MN} obtained ratios $5/3$ w.r.t. the cut LP, ratio $3/2$ w.r.t. to a strengthened ``leaf edge-cover'' LP,
and ratio $17/12$ not related to any LP. 
However, the analysis of \cite{MN} does not extend to the general {\sf TAP}.
In this paper, with the help of some ideas from \cite{MN,EFKN-TALG,KN-TAP,CGLS}, 
we introduce a new LP-relaxation and prove that its integrality gap at most $7/4$ for {\sf TAP}. 
This gives some hope to break the ratio $2$ for the weighted case.

Finally, we mention some work on the closely related {\sf $2$-Edge-Connected Subgraph} problem.
This problem was also vastly studied. For general weights, the best known ratio is $2$ by
Fredrickson and J\'{a}j\'{a} \cite{FJ}, which can also be achieved by the algorithms 
in \cite{KT} and \cite{Jain}. For particular cases, better ratios are known. 
Fredrickson and J\'{a}j\'{a} \cite{FJ-TSP} showed that when the edge weights
satisfy the triangle inequality, the Christofides heuristic has ratio $3/2$.
For the special case when all the edges of the input graph have unit weights
(the ``min-size'' version of the problem), the currently best known ratio is $4/3$ \cite{SV}. 

\section{A new LP-relaxation for {\sf TAP}}

In this section we introduce a new LP-relaxation for {\sf TAP},
and in subsequent section prove a $7/4$ integrality gap (for the unweighted case). 
This LP-relaxation combines ideas from \cite{MN,EFKN-TALG,KN-TAP,CGLS}, but also 
introducing new crucial valid constraints. 
Later, we will also introduce a technique of exchanging fractional values between the edges. 
We are not aware of such methods used to prove an integrality gap. 
We need some definition to introduce our LP.


For $u,v \in V$ let $(u,v)\in {\cal E}$ denote the edge in $T$ and 
$uv$ the link in $E$ between $u$ and $v$. 
A link $uv$ {\em covers} all the edges along the path $P(uv)$.
The choice of the root $r$ defines a partial order on $V$: $u$ is a {\em descendant} of $v$ 
(or $v$ is an {\em ancestor} of $u$) if $v$ belongs to $P(ru)$;
if, in addition, $(u,v) \in T$, then $u$ is a {\em child} of
$v$, and $v$ is the {\em parent} of $u$.  
The {\em leaves} of $T$ are the nodes in $V \setminus \{r\}$ that have no descendants. 
We denote the leaf set of $T$ by $L(T)$, or
simply by $L$, when the context is clear.  
The {\em rooted subtree} of $T$ induced by $r'$ and its descendants is
denoted by $T_{r'}$ ($r'$ is the root of $T_{r'}$). A subtree $T'$
of $T$ is called a {\em rooted subtree} of $T$ if $T'=T_{r'}$ for some $r' \in V$. 

\begin{definition} [shadow, shadow-minimal cover] \label{d:shadow}
Let $P(uv)$ denote the path between $u$ and $v$ in $T$.
A link $u'v'$ is a {\em shadow} of a link $uv$ if $P(u'v') \subseteq P(uv)$.
A cover $F$ of $T$ is {\em shadow-minimal} if for every link
$uv \in F$ replacing $uv$ by any proper shadow of $uv$ results in a set of links that does
not cover $T$.
\end{definition}

Every {\sf TAP} instance can be rendered closed under shadows by adding all shadows of existing links.  
We refer to the addition of all shadows as {\em shadow completion}.  
Shadow completion does not affect the optimal solution size, since every shadow can be replaced by some
link covering all edges covered by the shadow. Thus we may assume that
the set of links $E$ is closed under shadows.

\begin{definition} [twin-link, stem] \label{d:twin}
A link between two leaves $a,b$ of $T$ is a {\em twin-link} if its contraction results in a new leaf;
$a,b$ are called twins and their least common ancestor is called a {\em stem}.
Let $W$ denote the set of twin links, and for $e \in W$ let $s_e$ denote the stem of a twin-link $e$.
\end{definition}

For $A,B \subseteq V$ and $F \subseteq E$ let 
$\delta_F(A,B)$ denote the set of links in $F$ with one end in $A$ and the other end in $B$, and let
$\delta_F(A)=\delta_F(A,V \setminus A)$ denote the set of links in $F$ with exactly one endnode in $A$.
The default subscript in the above notation is $E$.
Let $L$ denote the set of leaves of $T$ and let 
$${\cal O}_L=\{A \subseteq V:|A \cap L| \mbox{ is odd}\} \ .$$
For a function $x$ on $E$ and $F \subseteq E$ let $x(F)=\sum_{e \in F}x_e$.

Let $\Pi$ be the polyhedron defined by the following set of linear constraints:
\begin{eqnarray}
x(\delta(A))       & \geq & 1                            \hphantom{aaaaaaaaaaaa} \forall A \in \hat{{\cal E}}                          \label{e:1} \\
x(\delta(A,V))        & \geq & \left\lceil |A \cap L|/2 \right\rceil   \hphantom{aaa} \forall A \in {\cal O}_L 
         \label{e:3} \\
x(\delta(v))       &   =  & 1                            \hphantom{aaaaaaaaaaaa} \forall v \in L                                 \label{e:2} \\
x_e-x(\delta(s_e)) &   =  & 0                            \hphantom{aaaaaaaaaaaa} \forall e \in W                                 \label{e:4} \\
x_e                & \geq & 0                            \hphantom{aaaaaaaaaaaa} \forall e \in E                                 \label{e:5}
\end{eqnarray} 
Inequalities (\ref{e:1}) and (\ref{e:5}) are the constraints of a standard LP-relaxation for {\sf TAP}.
Inequality (\ref{e:3}) were used in \cite{MN} to establish integrality gap of $1.5$ for
the special case of {\sf TAP} when every link connects two leaves.
We add over it the constraints (\ref{e:2}) and (\ref{e:4}), which are crucial to obtain integrality gap
better than $2$ for {\sf TAP}. 

Now we explain why the above LP is a relaxation for {\sf TAP}.

\begin{definition} [exact cover]
An edge set $F$ is an {\em exact cover} of $L$ if $|\delta_F(v)|=1$ for all $v \in L$.
\end{definition}

\begin{lemma} [\cite{KN-TAP}] \label{l:a} 
Given a {\sf TAP} instance with shadow completion, let $F$ be an optimal shadow-minimal solution 
with $|F \cap W|$ maximal. Then $F$ is an exact cover of $L$, and for any $e =ab \in W$, 
either $e \in F$ and $|\delta_F(s_e)|=1$, or $e \notin F$ and $|\delta_F(s_e)|=0$.
\qed
\end{lemma}

Let $\Pi_L$ be the polyhedron defined by (\ref{e:3}), (\ref{e:2}), and (\ref{e:5}).
Then $\Pi_L$ is the convex hull of the exact edge-covers of $L$, see \cite[Theorem~34.2]{Sch}; 
thus by Lemma~\ref{l:a}, these constraints are valid. 
The validity of the constraints (\ref{e:4}) follows also from Lemma~\ref{l:a}.
Consequently, the linear program $\tau=\min\{x(E):x \in \Pi\}$ is a rela\-xation for {\sf TAP}.
Combining techniques from \cite{EFKN-TALG,KN-TAP} and using some new methods, we prove the following.

\begin{theorem} \label{t:main}
{\sf TAP} admits a polynomial time algorithm that computes a solution $F$ such that $|F| \leq \frac{7}{4} \tau$.
\end{theorem}


\section{Proof of Theorem~\ref{t:main}} \label{s:main}


\subsection{Reduction to the minimum weight leaf edge-cover problem} 

Let $S=\{s_e:e \in W\}$ be the set of stems of $T$ and let $R=V\setminus(L \cup S)$. 
Let $\rho \geq 1.5$ be a parameter, set later to $\rho=7/4$. 
Define a weight function $w$ on $E(L,V)$ by:
$$ 
w_e = \left\{ 
\begin{array}{lll} 
\rho             \ \ \ \ \ \ \ \ & \mbox{ if } e \in \delta(L,L) \setminus W  \\
\rho-\frac{1}{2}                 & \mbox{ if } e \in \delta(L,V \setminus L)  \\
\rho+\frac{1}{2}                 & \mbox{ if } e \in W 
\end{array} 
\right . 
$$

\begin{lemma} \label{l:coupons}
Let $F_L$ be a minimum weight exact cover of $L$ and $x \in \Pi$ such that $x(E)=\tau$. Then:
\begin{equation} \label{e:coupons}
\rho \tau \geq w(F_L)+\frac{1}{2}\sum_{v \in R} x(\delta(v)).
\end{equation}
\end{lemma}
\begin{proof}
Let $x'$ be defined by $x'_e=x_e$ if $e \in \delta(L,V)$ and $x'_e=0$ otherwise.
Note that $x' \in \Pi_L$, since $x$ satisfies (\ref{e:3}), (\ref{e:2}), and (\ref{e:5}). 
Since $F_L$ is an optimal (integral) exact cover of $L$ with respect to the 
weights $w_e$ and $x' \in \Pi_L$, we have:
$$x'\cdot w \geq w(F_L) \ .$$ 
Thus to prove the lemma, it is sufficient to prove the following:
$$\rho \tau \geq x' \cdot w + \frac{1}{2}\sum_{v \in R} x(\delta(v)) \ .$$

Assign $\rho x_e$ tokens to every $e \in E$. The total amount of tokens is exactly $\rho x(E)= \rho \tau$.
We will show that these tokens can be moved around such that the following holds:
\begin{itemize}
\item[(i)]
Every $e \in \delta(L,L)$, and thus every $e \in W$, keeps its initial $\rho x_e$ tokens.
\item[(ii)]
Every $e \in \delta(L,V \setminus L)$ keeps $(\rho-\frac{1}{2})x_e$ tokens from its initial $\rho x_e$ tokens.
\item[(iii)]
Every $e \in W$ gets additional $\frac{1}{2}x_e$ tokens, to a total of $(\rho+\frac{1}{2})x_e$ tokens.
\item[(iv)]
Every $v \in R$ gets $\frac{1}{2}x_e$ tokens for each $e \in \delta(v)$.
\end{itemize}
This distribution of tokens is achieved in two steps.
In the first step, for every $e \in E$, move $\frac{1}{2}x_e$ token from the $\rho x_e$ tokens of $e$ 
to each non-leaf endnode of $e$, if any. 
Note that after this step, (i) and (ii) hold,
and every $v \in V \setminus L$ gets $\frac{1}{2}x_e$ token for each $e \in \delta(v)$.
In the second step, every $e \in W$ gets all the tokens moved at the first step to its stem $s_e$
by the links in $\delta(s_e)$.
The amount of such tokens is $\frac{1}{2}x(\delta(s_e))=\frac{1}{2}x_e$, by (\ref{e:4}).
This gives an assignment of tokens as claimed, and thus the proof of the lemma is complete.
\qed
\end{proof}

To prove Theorem~\ref{t:main} we prove the following.

\begin{theorem} \label{t:main+}
For $\rho=7/4$ there exist a polynomial time algorithm  
that given an instance of {\sf TAP} computes a solution $I$ 
of size at most the right-hand size of {\em (\ref{e:coupons})}.
Thus $|I| \leq \frac{7}{4} \tau$.
\end{theorem}

\subsection{The algorithm (Proof of Theorem~\ref{t:main+})}

The algorithm as in Theorem~\ref{t:main+} follows the line of the algorithms in \cite{EFKN-TALG,KN-TAP},
with the major difference that the lower bound is compared versus a fractional solution $x$ and 
not versus an integral solution. This makes the analysis much more involved.

Let $F_L$ be a minimum $w$-weight exact cover of $L$ 
and let $M=\delta_{F_L}(L,L)$ be the set of leaf-to-leaf links in $F_L$.
Initially, assign tokens to the links in $M$ and to the nodes as follows:

\medskip
\medskip

\noindent
{\bf Initial token assignment.} 
\begin{itemize}
\item[(i)]
Every link in $M \setminus W$ gets $\rho$ tokens. 
\item[(ii)]
Every link in $M \cap W$ gets $\rho+\frac{1}{2}$ tokens.
\item[(iii)]
Every leaf unmatched by $M$ gets $\rho-\frac{1}{2}$ tokens.
\item[(iv)]
Every $v \in R$ gets $\frac{1}{2}x(\delta(v))$ tokens.  
\item[(v)]
The root $r$ gets $1$ token.
\end{itemize}

\medskip

By Lemma~\ref{l:coupons} we have: 

\begin{corollary}
The total amount of tokens is at most $\rho \tau+1$.
\end{corollary} 

At each iteration, the algorithm iteratively contracts a certain subtree $T'$ of $T/I$, which means the following: 
combine all nodes in $T'$ into a single node $v$, delete edges and links with both endpoints in $T'$,
and the edges and links with one endpoint in $T'$ now have $v$ as their new endpoint.
For a set of links $I \subseteq E$, let $T/I$ denote the tree obtained by
contracting every 2-edge-connected component of $T \cup I$ into a
single node.  Since all contractions are induced by subsets of links,
we refer to the contraction of every $2$-edge-connected component of 
$T \cup I$ into a single node simply as the contraction of the links in $I$.

We refer to the nodes created by contractions as {\em compound nodes}.
Each compound node always owns $1$ token.
Non-compound nodes are referred to as {\em original nodes} (of $T$). 
Each time a contraction takes place, the new compound node gets $1$ token,
which together with the links added is paid by the tokens owned by the contracted subtree.
For technical reasons, $r$ is also considered as a compound node.
The non-contracted links in $M$ and original nodes keep their tokens.
This means that the algorithm maintains the following invariant:

\medskip
\medskip

\noindent
{\bf Tokens Invariant.} 
\begin{itemize}
\item[(i)]
Every link in $M \setminus W$ owns $\rho$ tokens. 
\item[(ii)]
Every link in $M \cap W$ owns $\rho+\frac{1}{2}$ tokens.
\item[(iii)]
Every leaf unmatched by $M$ owns $\rho-\frac{1}{2}$ tokens.
\item[(iv)]
Every $v \in R$ owns $\frac{1}{2}x(\delta(v))$ tokens.  
\item[(v)]
Every compound node owns $1$ token.
\end{itemize}

\medskip

The algorithm starts with a partial solution $I=\emptyset$ and with 
$credit(T/I)=credit(T)$ being the right-hand side of (\ref{e:coupons}) plus $1$. 
It iteratively finds a subtree $T'$ of $T/I$ and a cover $I'$ of $T'$,
and {\em contracts $T'$ with $I'$}, which means the following:
add $I'$ to $I$, contract $T'$, and assign $1$ token to the new compound node. 
To use the notation $T/I$ properly, we will assume that $I$ is an exact cover of $T'$,
namely, that the set of edges of $T/I$ that is covered by $I$ equals the set of edges
of $T'$ (this is possible due to shadow completion). 

\begin{definition}
A contraction of $T'$ with $I'$ is {\em legal} if $tokens(T') \geq |I'|+1$. 
\end{definition}

This means that the set $I'$Œ of links added
to $I$ and the $1$ token assigned to the new compound node are paid by the total
amount of tokens in $T'$. The tokens owned by $T'$ are not reused in any other way.
We do only legal contractions, which implies that at any step of the algorithm
$$
|I|+tokens(T/I) \leq credit(T) \ .
$$
Thus at the last iteration, when $T/I$ becomes a single compound node, $|I|$ is at
most the right-hand side of (\ref{e:coupons}).

Another property of our contracted tree $T'$ is given in the following definition.

\begin{definition} [$M$-compatible tree] 
Let $M$ be a matching on the leaves of $T/I$.
A subtree $T'$ of $T/I$ is {\em $M$-compatible} if for any $bb' \in M$ either 
both $b,b'$ belong to $T'$, or none of $b,b'$ belongs to $T'$.
We say that a contraction of $T'$ is $M$-compatible if $T'$ is $M$-compatible.
\end{definition}

A simple example of a legal $M$-compatible contraction is as follows.

\begin{definition}
Let $uv \in E$ such that 
both $u,v$ are unmatched by $M$.
Then adding $uv$ to the partial solution $I$ 
and assigning $1$ token to the obtained compound node is called a {\em greedy contraction}.
\end{definition}

Note that a greedy contraction is always $M$-compatible, and that if $\rho \geq 1.5$ then it is legal.
One of the steps of the algorithm is to apply greedy contractions exhaustively;
clearly, this can be done in polynomial time.

We now describe a more complicated type of legal $M$-compatible contractions. 
The {\em up-link} $up(a)$ of a node $a$ is the link $au$ such that 
$u$ is as close as possible to the root; such $u$ is called the {\em up-node} of $a$. 
Assuming shadow completion, such $u$ is unique and is an ancestor of $a$. 
For a rooted subtree $T'$ of $T/I$ and a node $a \in T'$ we say that
$T'$ is {\em $a$-closed} if the up-node of $a$ belongs to $T'$ 
(namely, if no link incident to $a$ has its other endnode outside $T'$), 
and $T'$ is {\em $a$-open} otherwise.
For a node set $U$ we let $up(U)=\{up(u):u \in U\}$. A rooted subtree $T'$ of $T/I$ is 
{\em $U$-closed} if there is no link in $E$ from $U \cap T'$ to $T \setminus T'$. 
$T'$ is {\em leaf-closed} if it is $L(T)$-closed. A leaf-closed $T'$ is
{\em minimally leaf-closed} if any proper rooted subtree of $T'$ is not leaf-closed.

\begin{definition} [semi-closed tree] 
Let $M$ be a matching on the leaves of $T/I$.
A rooted subtree $T'$ of $T/I$ is {\em semi-closed} 
(w.r.t. $M$) if it is $M$-compatible and closed w.r.t. its unmatched leaves.
$T'$ is {\em minimally semi-closed} if $T'$ is semi-closed but any proper subtree of $T'$ is not semi-closed.
\end{definition}

For a semi-closed tree $T'$Œ let us use the following notation:
\begin{itemize}
\item
$M'=M(T')$ is the set of links in $M$ with both endnodes in $T'$Œ.
\item
$U'=U(T')$ is the set of unmatched leaves of $T'$.
\end{itemize}

\begin{lemma} [\cite{EFKN-TALG,KN-TAP}] \label{l:up}
If $T'$ is minimally semi-closed then $M' \cup up(U')$ is an exact cover of $T'$.
\end{lemma}

Thus a minimally semi-closed tree $T'$ admits a cover of size $|M'|+|U'|$. 
This motivates the following definition, that concerns {\em arbitrary} semi-closed trees, 
that may not be minimal.

\begin{definition} [deficient tree]
A semi-closed tree $T'$ is {\em deficient} if \\
$credit(T')< |M'|+|U'|+1$.
\end{definition}

Summarizing, our algorithm maintains the following additional invariant:

\medskip
\medskip

\noindent
{\bf Partial Solution Invariant.} \\
The partial solution $I$ is obtained by sequentially applying a greedy contraction 
or a legal semi-closed tree contraction with an exact cover.

\medskip

In the next section we will prove the following key statement.

\begin{lemma} \label{l:B'}
Suppose that $\rho=7/4$ and that the Tokens Invariant and the Partial Solution Invariant hold for $T$, $M$, and $I$, 
and that $T/I$ has no greedy contraction. 
Then there exists a polynomial time algorithm that finds a non-deficient semi-closed tree $T'$ of $T/I$ 
and an exact cover $I \subseteq E$ of $T'$ of size $|I'|=|M'|+|U'|$. 
Furthermore, such $T'$ can be found even without knowing the LP solution $x$.
\end{lemma}


\begin{algorithm}[H]
\caption{{\sc LP-Tree-Cover}$(T=(V,{\cal E}),E)$  (A $1.75$-approximation algorithm)} \label{alg:F}
{\bf initialize:}  $I \gets \emptyset$  \\
$F_L \gets$ minimum $w$-weight exact edge cover of $L$, $M \gets \delta_{F_L}(L,L)$.  \\
\While{\em $T/I$ has more than one node}
{
Exhaust greedy contractions and update $I$ and $M$ accordingly. \\
Find a subtree $T'$ of $T/I$ and an exact cover $I'$ of $T'$ as in Lemma~\ref{l:B'}. \\
Contract $T'$ with $I'$.  
}
\Return{$I$}
\end{algorithm}

\vspace{0.2cm}

Algorithm {\sc LP-Tree-Cover} (Algorithm~\ref{alg:F}) initiates $I \gets \emptyset$ as a partial cover.
It computes a minimum $w$-weight exact edge cover $F_L$ of $L$, sets $M=\delta(F_L,F_L)$ 
and initiates the described credit scheme.
In the main loop, the algorithm iteratively exhausts greedy contractions,
then computes $T',I'$ as in Lemma~\ref{l:B'}, and contracts $T'$ with $I'$.
The stopping condition is when $I$ covers $T$, namely, when $T/I$ is a single node.

It is easy to see that all the steps in the algorithm can be implemented in polynomial time.
The credit scheme used implies that the algorithm computes a 
solution $I$ of size at most $\rho$ times the right-hand size of~(\ref{e:coupons}).
Hence it only remains to prove Lemma~\ref{l:B'}.

\section{Proof of Lemma~\ref{l:B'}} \label{s:B'}

In what follows, for a subtree $T'$ of $T/I$ let us use the following notation:
\begin{itemize}
\item 
$M'=M(T')$ is the set of links in $M$ with both endnodes in $T'$.
\item
$U'$ is the set of leaves of $T'$ unmatched by $M$ 
\item
$U'_0$ is the set of the original leaves in $U'$. 
\item
$L'=L(T')$ is the set of leaves of $T'$ and
$S'=S(T')$ is the set of stems of $T'$.
\item
$R'=V(T') \setminus (L' \cup S')$ and $\Sigma=\sum_{v \in R'}x(\delta(v))$.
\item
$C'$ is the set of non-leaf compound nodes of $T'$ (this includes $r$, if $r \in T'$).
\end{itemize}

\begin{definition} [dangerous tree] \label{d:dangerous}
A semi-closed tree $T'$ is called {\em dangerous} (see Fig.~\ref{f:dang}) if 
$|C'|=|S'|=|U'_0|=0$, $|M'|=1$, $|L'|=3$, and if $a$ is the unmatched leaf of $T$ then 
there exists an ordering $b,b'$ of the matched leaves of $T'$ such that $ab' \in E$, 
the contraction of $ab'$ does not create a new leaf, and $T'$ is $b$-open.
If such an ordering $b,b'$ is not unique 
(namely, if also $ab \in E$, the contraction of $ab$ does not create a new leaf, and $T'$ is $b'$-open,
see Fig.~\ref{f:dang}(b)), then we will assume that the up-node of $b$ is an ancestor of
the up-node of $b'$.
\end{definition}

\begin{figure}
\centering 
\epsfbox{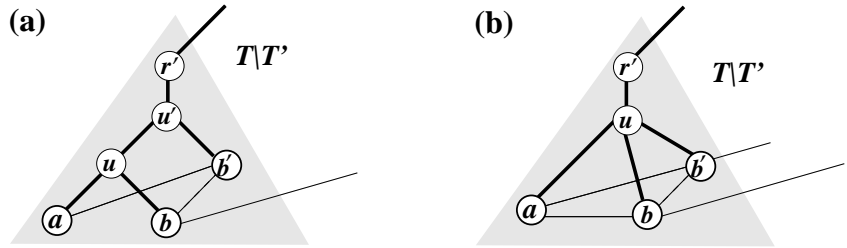}
\caption{Dangerous trees. The dashed arc shows the matched pair $bb'$. 
Solid thin lines show links that must exist in $E$. Here $u$ is not a stem and $a$ is a compound node. 
Some of the edges of $T$ can be paths and $u=r'$ may hold.} 
\label{f:dang}
\end{figure}

In the next section we will prove the following key statement.

\begin{lemma} \label{l:char}
Under the assumptions of Lemma~\ref{l:B'}, any deficient tree is dangerous.
\end{lemma}

Note that the property of a tree being dangerous depends only on the structure
of the tree and existence/absence of certain links in $E$ and $M$, and thus can be
tested in polynomial time.
If $T/I$ has a minimally semi-closed subtree $T'$ that is not dangerous, then $T'$ is not deficient,
so $T'$ and $I'=M' \cup up(U')$ satisfy the requirement of Lemma~\ref{l:B'}.
Otherwise, if all minimally semi-closed subtrees of $T/I$ are dangerous, then
we execute Algorithm~\ref{alg:tF} below.
In the algorithm we define a new matching $\tilde{M}$ obtained
from $M$ by replacing the link $bb'$ by the link $ab'$ in each dangerous tree.
Note that the property of a tree being semi-closed or dangerous depends on the
matching. In what follows, ``dangerous'' always means w.r.t. the matching $M$; 
for ``semi-closed'' the default matching is $M$, and we will specify each time when a tree
is semi-closed w.r.t. the new matching $\tilde{M}$.

\vspace{0.2cm}

\begin{algorithm}[H]
\caption{{\sc Find-Tree}$(T=(V,{\cal E}),E,M)$ (Finds a non-dangerous semi-closed tree $T'$ and 
its exact cover $I'$ of size $|I'|=|M'|+|U'|$, when all minimally semi-closed trees are dangerous.)}
\label{alg:tF}
Let $\tilde{M}$ be a matching on the leaves of $T/I$ obtained from $M$ by replacing 
the link $e=bb'$ by the link $\tilde{e}=ab'$ in every dangerous tree $T_0$ of $T$. \\
Let $T'$ be a minimally semi-closed tree w.r.t. the matching $\tilde{M}$. \\
\Return{$T'$ and $I'=\tilde{M}(T') \cup up(\tilde{U}')$}, 
where $\tilde{U}'$ is the set of unmatched leaves of $T'$ w.r.t. $\tilde{M}$.
\end{algorithm}

\vspace{0.2cm}

\begin{lemma}
Algorithm~\ref{alg:tF} finds a non-dangerous (non-minimal) 
semi-closed tree $T'$ and its cover $I'$ of size $|I'|=|M'|+|U'|$.
\end{lemma}
\begin{proof}
The statement was essentially proved in \cite{EFKN-TALG,KN-TAP}, as 
our dangerous trees coincide with a similar definition in \cite{EFKN-TALG,KN-TAP}.
The key point is that for any dangerous tree $T_0$, 
either all of $a,b,b'$ belong to $T'$, or none of them belongs to $T'$.
This implies that $T'$ is semi-closed and properly contains some dangerous tree. 
Consequently, it can be shown that $T'$ is not dangerous.
Furthermore, the bijective correspondence between links in $M$ and $\tilde{M}$
implies $|M(T')|=|\tilde{M}(T')|$, and the statement follows.
\qed
\end{proof}

The proof of Lemma~\ref{l:B'} is now complete. It remains only to prove Lemma~\ref{l:char}.

\section{Proof of Lemma~\ref{l:char}}

Let $T'$ be a deficient tree with root $r'$, so $tokens(T')-(|M'|+|U'|)<1$.
We will show that $T'$ is dangerous (under the assumptions of Lemma~\ref{l:B'}).

In what follows, note that $M$ is a matching on the leaves of $T$, 
and by the Partial Solution Invariant, 
$M$ remains a matching on the leaves of $T/I$ and every leaf of $T/I$ matched
by $M$ is an original leaf of $T$; this is so since we contract only $M$-compatible trees. 
In particular, $x(\delta(b))=1$ for every $b \in L \setminus U$, by (\ref{e:2})
(note that $L \setminus U$ is the set of leaves of $T/I$ matched by $M$). 
Furthermore, the Partial Solution Invariant implies that every stem $s$ in $T/I$ 
has exactly two leaf descendant, and they are both original leaves.


The amount of tokens owned by $T'$ is: 
$$
tokens(T')=\rho|M'|+\frac{1}{2}|M' \cap W|+|U'|+(\rho-\frac{1}{2})|U'_0|+|C'|+\Sigma
$$ 
Thus
$$
tokens(T')-(|M'|+|U'|)=(\rho-1)|M'|+\frac{1}{2}|M' \cap W|+(\rho-\frac{3}{2})|U'_0|+|C'|+\Sigma \ .
$$ 
Since $T'$ is deficient, $tokens(T')-(|M'|+|U'|)<1$. Thus for $\rho=7/4$ we have:
\begin{equation} \label{e:u0}
tokens(T')-(|M'|+|U'|)=\frac{7}{4}|M'|+\frac{1}{2}|M' \cap W|+\frac{1}{4}|U'_0|+|C'|+\Sigma < 1 \ .
\end{equation} 

\begin{lemma} \label{l:CMS}
$|C'|=0$, $|M' \cap W|=0$, $|M'| \leq 1$, and $|S'|=0$.
\end{lemma}
\begin{proof}
From (\ref{e:u0}) we immediately get that $|C'|=0$, $|M' \cap W|=0$, 
$|M'| \leq 1$, and if $|M'|=1$ then $|U'_0|=0$. 
It remains to prove that $|S'|=0$. 
Suppose to the contrary that $T'$ has a stem $s$.
Let $a,b$ be the two leaf descendants of $s$, so $a,b$ are original leaves and $ab $ is a twin link.
Since $ab \in W$, $ab \notin M'$.
From the assumption that that $T/I$ has no link greedy
contraction we get that one of $a,b$ is matched by $M$,
as otherwise $ab$ gives a greedy contraction.
Moreover, $|M' \cap W|=0$ and $|M'| \leq 1$ implies that exactly one of $a,b$ is matched by $M$.
This implies the contradiction $|M'|=1$ and $|U'_0| \geq 1$. 
\qed
\end{proof}

\begin{lemma} \label{l:sigma}
$\Sigma \geq |U'|+1-2|M'|$. 
\end{lemma}
\begin{proof}
Note that no link has both endnodes in $U'$ (since $T/I$ has no greedy contraction),
and that $\delta(U') \cap \delta(T') = \emptyset$ (since $T'$ is $U'$-closed).
Thus
$$
x(\delta(U') \cup \delta(T')) = \sum_{u \in U'} x(\delta(u)) + x(\delta(T')) \geq |U'|+1 \ .
$$
If $e \in \delta(U')$, then $e$ contributes $x_e$ to $\Sigma$, unless $e$ is incident to a matched leaf.
However, $x(\delta(b))=1$ for every matched leaf $b$, 
and the number of matched leaves in $T'$ is exactly $2|M'|$.
Hence $\Sigma \geq x(E') \geq |U'|+1-2|M'|$, as claimed.
\qed
\end{proof}

\begin{lemma}
$|M'|=1$, $\Sigma<1/2$, and $|L'|=3$. 
\end{lemma}
\begin{proof}
We prove all statements by contradiction. 
If $|M'| \neq 1$, then $|M'|=0$, by Lemma~\ref{l:CMS}. 
This implies $\Sigma \geq |U'|+1 \geq 2$, by Lemma~\ref{l:sigma}, 
and we obtain the contradiction 
$tokens(T')-(|M'|+|U'|) \geq \frac{1}{2} \Sigma \geq 1$. Thus $|M'|=1$.

If $\Sigma \geq \frac{1}{2}$, then by (\ref{e:u0}) we get
$tokens(T')-(|M'|+|U'|) \geq \frac{7}{4}|M'|+\frac{1}{2}\Sigma \geq 1$,
contradicting the assumption that $T'$ is deficient.

We show that $|L'|=3$. 
Note that $|L'| \geq 4$ is not possible, since then $|U'| \geq 2$, which implies the contradiction 
$\Sigma \geq |U'|+1-2|M'| \geq 1$. 
Also, $|L'|=1$ is not possible, since then $|M'|=0$.
We are therefore left with the case $|L'|=2$, say $L'=\{b,b'\}$.
Then, since $|M'|=1$, we have $M'=\{bb'\}$.  
Consequently, the contraction of $bb'$ creates a new leaf.
We obtain a contradiction by showing that then the path between $b$ and $b'$ in $T/I$ 
has an internal compound node. By the Partial Solution Invariant $b,b'$ are original leaves. 
Note that in the original tree $T$ the contraction of $bb'$ does not create a new leaf, 
since $bb' \notin W$. This implies that in $T$, there is a subtree $\hat{T}$ of $T$ hanging
out of a node $z$ on the path between $b$ and $b'$ in $T$. 
This subtree $\hat{T}$ is not present in $T/I$, hence it was contracted into a compound node 
during the construction of our partial solution $I$. 
Thus $T/I$ contains a compound node $\hat{z}$ that contains $\hat{T}$,
and since $\hat{z}$ contains a node $z$ that belongs to the path between $b$ and $b'$ in $T$, 
the compound node of $T/I$ that contains $z$ belongs to the path between $b$ and $b'$
in $T/I$, and it is distinct from $b,b'$, since both $b,b'$ are original leaves.
\qed
\end{proof}

\begin{figure}
\centering 
\epsfbox{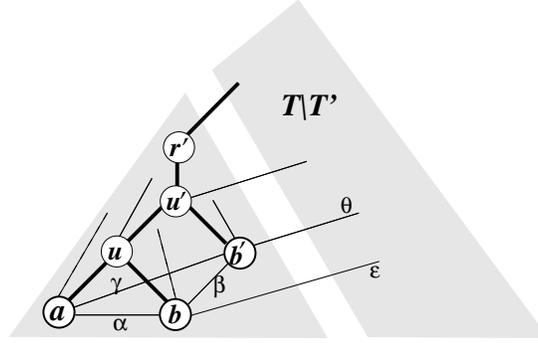}
\caption{Illustration to the proof of Lemma~\ref{l:char} (links in $E_1$ are shown by dashed lines).} 
\label{f:count}
\end{figure}

Now we finish the proof of Lemma~\ref{l:char}. 
Let $bb'$ be the matched pair and $a$ the unmatched leaf of $T'$.
Let $u$ and $u'$ be the least common ancestor of $ab$ and $ab'$, respectively,
and assume w.l.o.g. that $u$ is a descendant of $u'$ (see Fig.~\ref{f:count}).
Let $x_{ab}=\alpha$, $x_{bb'}=\beta$, $x_{ab'}=\gamma$, $x(\delta(b,T \setminus T')=\epsilon$,
and $x(\delta(b',T \setminus T')=\theta$.
To finish the proof of the lemma, it is sufficient to show the following: 
\begin{itemize}
\item
If $u \neq u'$, then $\gamma>0$ and $\epsilon>0$. 
\item
If $u=u'$, then at least one of the following holds: $\gamma,\epsilon>0$ or $\alpha,\theta>0$.
\end{itemize}
Indeed, if $u\neq u'$, then $\gamma>0$ implies that the link $ab'$ exists, 
and $\epsilon>0$ implies that $T'$ is $b$-open. 
Thus, by the definition, $T'$ is dangerous.
The same holds if $u=u'$ and $\gamma,\epsilon>0$. 
If $u=u'$ and $\alpha,\theta>0$, then $ab$ exists (since $\alpha>0$) and $T'$ is $b'$-open (since $\theta>0$);
thus by exchanging the roles of $b,b'$ we get that $T'$ is dangerous, by the definition. 
Consider the following links sets: 
\begin{itemize}
\item
$E_1=\delta(\{a,b\},R')=x(\delta(a,R'))+x(\delta(b,R'))$ are the links from $a,b$ to $R'$.
\item
$E_2=\delta(T_u \setminus \{a,b\},T \setminus T_u)=\delta(T_u \cap R',T \setminus T_u)$
are the links from $T_u \cap R'$ to nodes outside $T_u$.
\item
$E_3=\delta(R',T \setminus T')$ are the links from $R'$ to nodes outside $T'$.
\end{itemize}
Any $e \in E_1 \cup E_2 \cup E_3$ contributes $x_e$ to $\Sigma$.
Recalling that $x(\delta(b))=x(\delta(b'))=1$ and $x(\delta(a)) \geq1$, it is easy to verify that:
\begin{itemize}
\item[(i)]
$\Sigma \geq x(E_1) +x(E_2) = x(\delta(T_u)) -\gamma \geq 1-\gamma$ if $u \neq u'$. \\
The first inequality is since 
every $e \in E_1 \cup E_2$ contributes $x_e$ to $\Sigma$ and since $E_1 \cap E_2=\emptyset$.
\item[(ii)]
$\Sigma \geq x(E_1) = x(\delta(a,R'))+x(\delta(b,R')) \geq 
x(\delta(a,R')) \geq x(\delta(a))-\alpha-\gamma \geq 1-\alpha-\gamma$. 
\item[(iii)]
$\Sigma \geq x(E_3) \geq x(\delta(T',T \setminus T'))-\theta-\epsilon \geq 1-\theta-\epsilon$.
\item[(iv)]
$\alpha+\epsilon \leq 1$ (since $x(\delta(b))=1$) and 
$\gamma+\theta \leq 1$   (since $x(\delta(b'))=1$).  
\end{itemize}

Suppose that $u \neq u'$. Then $1/2 > \Sigma \geq 1-\gamma$, by (i); hence $\gamma>1/2$.
If $\epsilon=0$ then (iii) implies $\theta>1/2$, and we obtain a contradiction to (iv)
$\gamma+\theta>1$. Thus, $\gamma,\epsilon>0$ in this case, as claimed.

Suppose that $u=u'$. 
By (ii) and (iii), $\alpha+\gamma>1/2$ and $\theta+\epsilon>1/2$.
One can easily verify that combined with (iv) this gives that we must have 
$\gamma>0$ and $\epsilon>0$, or $\alpha>0$ and $\theta>0$, as claimed.

This concludes the proof of Lemma~\ref{l:char}.



\end{document}